\documentclass[letterpaper, 10 pt, conference]{ieeeconf}
\usepackage{lipsum}
\usepackage{dsfont}
\usepackage{cite}
\usepackage{amsmath,amssymb,amsfonts}
\usepackage{algorithmic}
\usepackage{graphicx}
\usepackage{textcomp}
\usepackage{xcolor}
\def\BibTeX{{\rm B\kern-.05em{\sc i\kern-.025em b}\kern-.08em
    T\kern-.1667em\lower.7ex\hbox{E}\kern-.125emX}}

\newcommand{\bE} { {\mathbb E}}
\newcommand{\bP} { {\mathbb P}}
\newcommand{\bR} { {\mathbb R}}
\newcommand{\bZ} { {\mathbb Z}}
\newcommand{\bS} { {\mathbb S}}
\newcommand{\kl} { {\mathrm {KL}}}

\newcommand{\cW} { {\mathcal W}}

\newcommand{\cF} { {\mathcal F}}

\newcommand{\cN} { {\mathcal N}}

\newcommand{\tr} { {\operatorname{Tr}}}

\newtheorem{secthm}{Theorem}[section]
\newtheorem{seccor}[secthm]{Corollary}
\newtheorem{seclem}[secthm]{Lemma}

\newtheorem{secprob}[secthm]{Problem}
\newtheorem{secdefn}[secthm]{Definition}
\newtheorem{secrem}[secthm]{Remark}

\def\red{\hfill $\lhd$}
\begin{document}

\IEEEoverridecommandlockouts 
\overrideIEEEmargins

\title{\LARGE \bf
Privacy Preservation for Statistical Input in Dynamical Systems\\
}

\author{
    Le Liu, Yu Kawano, and Ming Cao
    \\
    \thanks{The work of Liu and Cao was supported in part by the Netherlands Organization for Scientific Research (NWO-Vici-19902). The work of Liu is partly supported by Chinese Scholarship Council. The work of Yu Kawano was supported in part by JSPS KAKENHI Grant Number JP22KK0155.}
\thanks{Le Liu and Ming Cao are with the Faculty of Science and Engineering, Univerisity of Groningen, 9747 AG Groningen, The Netherlands {\tt \small \{le.liu, m.cao\}@rug.nl}}
    \thanks{Yu Kawano is with the Graduate School of Advance Science and Engineering, Hiroshima Univeristy, Higashi-hiroshima 739-8527, Japan 
        {\tt\small ykawano@hiroshima-u.ac.ip}}
}

\maketitle

\begin{abstract}
    This paper addresses the challenge of privacy preservation for statistical inputs in dynamical systems. Motivated by an autonomous building application, we formulate a privacy preservation problem for statistical inputs in linear time-invariant systems. What makes this problem widely applicable is that the inputs, rather than being assumed to be deterministic, follow a probability distribution, inherently embedding privacy-sensitive information that requires protection. This formulation also presents a technical challenge as conventional differential privacy mechanisms are not directly applicable. Through rigorous analysis, we develop strategy to achieve $(0, \delta)$ differential privacy through adding noise. Finally, the effectiveness of our methods is demonstrated by revisiting the autonomous building application.
\end{abstract}


\section{Introduction}
The rapid development of the Internet of Things (IoT) and cloud computing technologies has revolutionized industrial automation, enabling networked control systems to perform more efficiently \cite{gupta2009networked}. In such systems, sensors serve as critical components, continuously collecting data to support real-time control, diagnostics, and performance monitoring. However, the widespread use of sensor data raises privacy concerns, especially when sensitive measurements are shared with untrusted third parties. For example, sensor readings, such as carbon dioxide ($\mathrm{CO}_2$) levels and temperature, can reveal the parameters of statistical inputs, thereby exposing the occupancy level of a building\cite{ebadat2015regularized}. These vulnerabilities highlight the increasing conflict between data utility and privacy of statistical inputs in networked control systems. As a result, developing privacy-preserving sensor data sharing solutions for networked control systems has become a key area of research, combining fields like control theory, cyber-security techniques, and data science \cite{tourani2017security, cortes2016differential}.
\subsection{Literature Review}
In the field of privacy preservation in networked control systems, several related studies have been conducted, particularly in distributed optimization \cite{wang2022decentralized, wang2022quantization}, filtering \cite{le2013differentially}, and controller design \cite{kawano2020design, Yu2021, liu2024design}. Differential privacy has been employed to develop privacy-preserving filters, as demonstrated in \cite{le2013differentially}, where the authors also explored the relationship between the $H_{\infty}$ norm of a system and differential privacy.  

Distinct from \cite{le2013differentially}, the work in \cite{kawano2020design} investigated the interplay between differential privacy and input observability, highlighting that even a small amount of noise can significantly enhance privacy in systems with low observability. Following \cite{kawano2020design}, \cite{Yu2021} has generalized the results to protect privacy in modular control. Instead of introducing noise to system outputs, \cite{nekouei2022model} proposed a model randomization approach, where model parameters are randomized to preserve the privacy of the deterministic inputs. To preserve privacy for statistical inputs, privacy-preserving filters are designed using mutual information as a privacy metric \cite{bassi2020statistical}.

However, in the aforementioned privacy preservation approaches for differential privacy, privacy-sensitive information is typically assumed to be deterministic. In practical applications, such privacy information may instead correspond to statistical inputs. For instance, the $\mathrm{CO}_2$ input fluctuating with the occupancy level in a building can be represented as a stochastic input rather than a fixed value. In such cases, privacy protection using differential privacy should extend to the probability distribution of the input itself, an aspect not addressed in the existing literature.

\subsection{Contribution}
In this paper, we address the problem of privacy preservation for statistical inputs in dynamical systems. We begin by illustrating our motivation through a smart building application. Then, a general privacy preservation problem is formulated  for statistical inputs in linear time-invariant systems. To facilitate differential privacy analysis, we introduce an adjacency relationship that captures the statistical differences between the probability distributions of the inputs.

To design noise for privacy preservation, two key technical lemmas are established. The first lemma shows a relationship between the symmetrized Kullback–Leibler (KL) divergence and the Wasserstein distance in Gaussian distributions. The second lemma provides an upper bound on the Wasserstein distance between two outputs in terms of the Wasserstein distance between their corresponding inputs. Leveraging these lemmas along with Pinsker’s inequality, we derive a theorem that offers insights into preserving the privacy of probability distributions of the inputs.

Subsequently, we propose a general noise design strategy, followed by a specialized strategy for cases where the initial state is independent and follows a publicly known distribution. Finally, we revisit the motivating example to demonstrate the effectiveness of our proposed methods.
\subsection{Organization}
The remainder of this paper is organized as follows. In Section~\ref{sec:pre}, we provide a problem formulation and illustrate the generality of the proposed adjacency set. In Section~\ref{sec:dp}, we derive sufficient conditions for differential privacy and obtain noise adding strategies.  Section~\ref{sec:sim} provides a numerical examples to illustrate the proposed results. Finally, Section~\ref{sec:con} concludes the paper. All proofs are given in the Appendix.
\subsection{Notation}
The sets of real numbers is denoted by $\bR$. For $P \in \bR^{n \times n}$, $P \succ 0$ (resp. $P \succeq 0$) means that $P$ is symmetric and positive (resp. semi) definite. We also use $\mathbb{S}_{++}^n$ (resp. $\mathbb{S}_{+}^n$) to denote the set of positive (resp. semi) definite matrices. $\cN_n(\mu, \Sigma)$ represents a Gaussian distribution with the mean $\mu \in \bR^n$ and covariance $\Sigma \in \bS_{++}^{n}$. $\tr(\cdot)$ denotes the trace of a square matrix. A probability space is denoted by $(\Omega, \cF, \bP )$, where $\Omega$, $\cF$, and $\bP $ denote the sample space, $\sigma$-algebra, and probability measure, respectively. The expectation of a random variable is denoted by $\bE[\cdot]$. For $p \geq 1$ and Borel probability measures $\mu, \nu$ on $\bR^n$ with finite $p$-th moments, $p$-Wasserstein distance is defined by \cite{panaretos2019statistical}:

$$
\cW_p(\mu, \nu) = \left( \inf_{\gamma \in \Gamma(\mu, \nu)} \int_{\bR^n \times \bR^n} |x-y|_2^p \, \mathrm{d}\gamma(x, y) \right)^{1/p},
$$
where $\Gamma(\mu, \nu)$ denotes the set of all the couplings (joint distributions) of $\mu$ and $\nu$ with marginals $\mu$ and $\nu$.

For two multivariate Gaussian distributions $\bP_1 = \cN_n(m_1, \Gamma_1) \quad \text{and } \bP_2 = \quad \cN_n(m_2, \Gamma_2),
$
it holds that \cite{panaretos2019statistical}
\begin{align}
\label{eq:2wa}
\mathcal{W}_2^2(\bP_1, \bP_2) &= | m_1 - m_2 |^2_2 + \tr ( \Sigma_1 \nonumber \\
+& \Sigma_2 - 2 (\Sigma_1^{1/2} \Sigma_2 \Sigma_1^{1/2})^{1/2} ).
\end{align}

\section{Problem Formulation}
\label{sec:pre}
In this section, we present a formulation of the privacy preservation for statistical input in dynamical systems. While previous studies have considered deterministic inputs as privacy-sensitive information \cite{kawano2020design}, such an approach fails to capture the inherent randomness of real-world scenarios.  The next subsection presents a motivating example in which statistical input is regarded as natural privacy-sensitive information.
\subsection{A Motivating Example}
 Inspired by \cite{nekouei2022model}, we consider a building automation application in which a sensor measures the $\mathrm{CO}_2$ level in a room. The response of $\mathrm{CO}_2$ concentration to human presence can be modeled by:
\begin{align*}
x(t+1) & =a x(t)+ bu(t), \\
y(t) & =x(t)+v(t),
\end{align*}
where $a \in (0,1)$ and $b$ are constants, $x(t)$ denotes the $\mathrm{CO}_2$ level at time step $t$, $y(t)$ represents the sensor output at time step $t$. The contribution of individuals to the $\mathrm{CO}_2$ concentration in a building is inherently random due to variations in their activities and physiological factors. Specifically, \cite{rahman2021real} suggests that the contribution of an individual to the $\mathrm{CO}_2$ level can be modeled as a Gaussian random variable. Consequently, for a single individual, $u(t)$ is assumed to follow an independent and identical Gaussian distribution, i.e. $u(t) \sim \mathcal{N}_1(\bar{m} + m_u, \Sigma_u)$, where $\bar{m}$ reflects environmental influences, and $m_u$ and $\Sigma_u$ denote the individual emission rate and variability. For $n$ individuals in the building, the total contribution to the $\mathrm{CO}_2$ level can then be expressed as $u(t) \sim \mathcal{N}_1(\bar{m} + n m_u, n \Sigma_u)$, with the occupancy level directly influencing the mean and variance of the distribution.

In building automation, sensor measurements are commonly transmitted over communication networks for control and monitoring purposes. When the sensor measurements are accessible to untrusted third parties, the occupancy level, which is highly private information, can be inferred \cite{nekouei2022model}. Therefore, it is crucial to design the noise term $v(t)$ for privacy preservation of the probability distribution $\cN_1 (m, \Sigma)$, thereby protecting the privacy of the occupancy level. While previous research (e.g., \cite{kawano2020design, Yu2021}) has primarily focused on protecting deterministic inputs, it remains unclear how to ensure the privacy of a probability distribution.

\subsection{System Description}
Throughout this paper, we follow the convention by focusing on a finite datasets. In the context of dynamical systems, this corresponds to analyzing the system's properties within a finite time horizon.

Consider the following linear, discrete time, invariant system,
\begin{subequations}\label{eq:sys}
 \begin{align}
 \label{system_1}
	x(t+1) &= Ax(t)+Bu(t), \\
 \label{system_2}
    y(t) &= C x(t)+Du(t),
 \end{align}
\end{subequations}
where $x(t) \in \bR^{n}$ is the state, $u(t) \in \bR^m$ is the control input and $y(t)\in \bR^{q}$ is the measurement. All the matrices are of compatible dimensions. We define the initial state $x_0 := x(0)$.

For \eqref{eq:sys}, the output sequence $Y_t \in \bR^{(t+1)q}$ is described by
\begin{align}
    \label{eq:Yt}
    Y_t = O_tx_0 + N_t U_t,
\end{align}
where $O_t \in \bR^{(t+1)q \times n}$ and $N_t \in \bR^{(t+1)q \times (t+1)m}$ are 
\begin{align}
    \label{eq:Ot}
    O_t = \begin{bmatrix}
        C^{\top} & CA^{\top} & \dots & (CA^{t})^{\top}
    \end{bmatrix},
\end{align}
\begin{align}
    \label{eq:Nt}
    N_t:=\left[\begin{array}{ccccc}
D & 0 & \cdots & \cdots & 0 \\
C B & D & \ddots & & \vdots \\
C A B & C B & D & \ddots & \vdots \\
\vdots & \vdots & \ddots & \ddots & 0 \\
C A^{t-1} B & C A^{t-2} B & \cdots & C B & D
\end{array}\right].
\end{align}

To proceed with differential privacy analysis, we consider the output $y_{v}(t) = y(t) + v(t)$, where $v(t) \in \bR^{q}$ represents the added noise. From equation \eqref{eq:Yt}, $Y_{v,t} \in \bR^{(t+1)q}$ can be described by,
\begin{align}
    \label{eq:mech}
    Y_{v,t} = O_t x_0 + N_t U_t + V_t,
\end{align}
where $V_t$ is the noise to be designed.
Following the standard notations in differential privacy analysis, we call \eqref{eq:mech} a \emph{mechanism} \cite{dwork2006calibrating,dwork2006our}.
\subsection{Definition of Differential Privacy}
We are interested in evaluating differential privacy when the input data $(x_0, U_t)$ follows a certain probability distribution $\bP_{x_0} \times \bP_{U_t}$. To illustrate it clearly, we hereby define $\bP_{x_0} \times \bP_{U_t}$ as the probability distribution of the input.

Originally, differential privacy measures the privacy level of a mechanism based on the sensitivity of the output data $Y_{v,t}$ to the input data $(x_0, U_t)$. 
In our case, differential privacy measures the privacy level of a mechanism based on the sensitivity of the output data $Y_{v,t}$ to the probability distribution of the input $\bP_{x_0} \times \bP_{U_t}$. 
If two similar input pairs $(\bP_{x_0} \times \bP_{U_t}, \bP_{x_0'} \times \bP_{U_t'})$ result in significantly different outputs $(Y_{v,t}, Y_{v,t}')$, it indicates that the probability distribution of the input can be easily identified, suggesting weaker privacy. Thus, differential privacy is defined using pairs of probability distributions of inputs that satisfy specific adjacency relations. Using $p$-Wasserstein distance,
 the adjacency relationship is defined in the probability metric space, as given below.

\begin{secdefn}
\label{def:adj}
     Let $c>0$ and $p \in \mathbb{Z}_{+}$. Consider  Gaussian probability distributions $ \bP_{x_0} = \cN_{n}(m_{x_0}, \Sigma_{x_0}), \bP_{U_t} = \cN_{(t+1)m}(m_{U_t}, \Sigma_{U_t}), \bP_{x_0'} = \cN_{n}(m_{x_0'}, \Sigma_{x_0'}), \bP_{U_t'} = \cN_{(t+1)m}(m_{U_t'}, \Sigma_{U_t'})$. A pair of probability distribution of the inputs $ ((\bP_{x_0} \times \bP_{U_t}) , (\bP_{x_0^{\prime}} \times \bP_{U_t^{\prime}}))$ is said to belong to the binary relation $c$-adjacency under the $p$-Wasserstein distance if $\cW_{p}(\bP_{x_0} \times \bP_{U_t}, \bP_{x_0^{\prime}} \times \bP_{U_t^{\prime}}) \leq c$. The set of all pairs of the input probability Gaussian distribution that are $c$-adjacent under the $p$-Wasserstein distance is denoted by $\mathrm{Adj}_p^c$.
\end{secdefn}

In this paper, we specifically focus on Gaussian distributions as the probability distributions of inputs, as the Gaussian assumption is widely adopted and offers significant advantages in the analysis of linear systems.

Unlike \cite{kawano2020design}, we define the adjacency set using the $p$-Wasserstein distance. While other distance measures, such as the KL divergence and total variation, are also viable candidates, the Wasserstein distance offers a natural metric structure in the distribution space and is closely related to the Euclidean distance (see Remark \ref{rem:wa}). This is the primary reason for choosing the $p$-Wasserstein distance as the criteria to define adjacency relationship in this paper.

\begin{secrem}
\label{rem:wa}
    The adjacency relationship~\ref{def:adj} defined in this paper is a generalization of $\mathrm{Adj}_{2}^c$ in \cite{kawano2020design}. To illustrate this, consider the case where the inputs are deterministic data, i.e., $(\bP_{x_0} \times \bP_{U_t}) \times (\bP_{x_0^{\prime}} \times \bP_{U_t^{\prime}}) = (\delta_{x_0} \times \delta_{U_t}) \times (\delta_{x_0'} \times \delta_{U_t'})$, where $\delta_{x_0}$ denotes the Dirac delta function, which places all probability mass at $x_0$.
     The Dirac delta function can be viewed as a limiting case of a Gaussian distribution with variance tending to zero, i.e.,
     \begin{align*}
         \delta_{x_0}=\lim _{\sigma \rightarrow 0} \mathcal{N}\left(x_0, \sigma^2 I_n\right).
     \end{align*}
    Since the Wasserstein distance between two Dirac delta function is simply the Euclidean distance \cite{villani2009wasserstein}, we obtain
    \begin{align*}
        \cW_2(\bP_{x_0} \times \bP_{U_t}, \bP_{x_0^{\prime}} \times \bP_{U_t^{\prime}}) = |[x_0';U_t'] - [x_0';U_t']|_2,
    \end{align*}
    which exactly recovers $\mathrm{Adj}_2^p$ in \cite{kawano2020design}.
\end{secrem}

Now we are ready to define differential privacy of the mechanism \eqref{eq:mech}.

\begin{secdefn}
     Let $\left(\mathbb{R}^{(t+1) q}, \mathcal{F}, \mathbb{P}\right)$ be a probability space. The mechanism~\eqref{eq:mech} is said to be $(\varepsilon, \delta)$-differentially private for $\operatorname{Adj}_p^c$ at a finite time instant $t \in \mathbb{Z}_{+}$, if there exist $\varepsilon \geq 0$ and $\delta \geq 0$ such that

$$
\begin{aligned}
& \mathbb{P}\left(O_t x_0+N_t U_t+V_t \in \mathcal{S} \mid x_0 \sim \bP_{x_0}, U_t \sim \bP_{U_t} \right) \\
 \leq & \mathrm{e}^{\varepsilon} \mathbb{P}\left(O_t x_0^{\prime}+N_t U_t^{\prime}+V_t \in \mathcal{S} \mid x_0' \sim \bP_{x_0'}, U_t' \sim \bP_{U_t'}\right) \\
&+\delta, \quad \forall \mathcal{S} \in \mathcal{F}
\end{aligned}
$$

for any $((\bP_{x_0} \times \bP_{U_t}) , (\bP_{x_0^{\prime}} \times \bP_{U_t^{\prime}})) \in \operatorname{Adj}_p^c$.
\end{secdefn}

It is worth noting that although the setting in this paper is similar to previous differential privacy frameworks \cite{dwork2006calibrating, kawano2020design, wang2023differential}, our case presents more technical challenges. This is because, in the prior works \cite{kawano2020design, Yu2021}, only the mean of the final mechanism differs, which allows for the use of noise strategies developed in earlier literature (e.g., \cite{le2013differentially, balle2018improving}). In contrast, our formulation, which allows for differences in covariance, introduces additional complexities that prevent the direct application of existing noise-adding strategies.

\begin{secprob}
    In this paper, we are interested in answering the following problem.
    \begin{itemize}
        \item How can we design the noise to provide differential privacy for privacy preservation of probability distribution of inputs?
    \end{itemize} 
\end{secprob}

\section{Differential Privacy with Output Noise}
\label{sec:dp}
As clear from the definition, both $\varepsilon$ and $\delta$ depend on noise. In particular, we will demonstrate that the sensitivity of the dynamical system~\eqref{eq:sys} establishes a lower bound on the covariance matrix of the multivariate Gaussian noise required to achieve $(0, \delta)$-differential privacy. Henceforth, we refer to any mechanism in~\eqref{eq:mech} incorporating Gaussian noise as a Gaussian mechanism.

To proceed, we prove the following two lemmas, which are instrumental in proving Theorem~\ref{thm:dp}. The first lemma establishes an upper bound for the symmetrized KL divergence in terms of the 2-Wasserstein distance.
\begin{seclem}
\label{lem:disinq}
    For two Gaussian distributions $\bP_1 = \cN_{n}(m_1, \Sigma_1)$ and $\bP_2 = \cN_{n}(m_2, \Sigma_2)$, the symmetrized KL divergence and the Wasserstein distance between $\bP_1$ and $\bP_2$ satisfy the following inequality,
    \begin{align*}
    &\kl(\bP_1 \| \bP_2) +\kl(\bP_2 \| \bP_1) \\
    \leq & \frac{2}{\min\left(\lambda_{\min}\left(\Sigma_1\right), \lambda_{\min}\left(\Sigma_2\right)\right)} \cW_2^2(\bP_1, \bP_2),
\end{align*}
where $\kl(\bP_1 \| \bP_2)$ denotes the KL divergence between $\bP_1$ and $\bP_2$.
\end{seclem}

\begin{proof}
    The proof is provided in Appendix ~\ref{app1}.
\end{proof}

Since the output $Y_{v,t}$ is an affine transformation of $(x_0, U_t)$, we aim to understand the relationship between the 2-Wasserstein distance of the output and that of the input. The reasons will be explained later. The following lemma provides an upper bound on the 2-Wasserstein distance under an affine transformation, incorporating an additional scaling factor.
\begin{seclem}
\label{lem:aff}
    Suppose $ x \in \bR^n$ follows a probability distribution $ \bP_x $ and $ x' $ follows another probability distribution $ \bP_{x'} $. Let $ y = Fx + v $ and $ y' = Fx' + v $, where $F \in \bR^{m \times n}$ and $ v $ follows the distribution $ \bP_v $ independent of $ \bP_{x}$ and $ \bP_{x'}$. Then, the following holds:
    \begin{align*}
        \cW_2(\bP_y, \bP_{y'}) \leq |F|_2\cW_2(\bP_x, \bP_{x'}),
    \end{align*}
    where $\bP_y$ denotes the pushforward probability distribution of $y$, $\bP_{y'}$ denotes the pushforward probability distribution of $y'$, and $|F|_2$ is the induced $2$-norm of the matrix $F$.
\end{seclem}
\begin{proof}
    The proof is provided in Appendix~\ref{app2}.
\end{proof}

Thanks to Lemma \ref{lem:disinq} and Lemma \ref{lem:aff}, we can establish an upper bound for the symmetrized KL divergence between outputs using the 2-Wasserstein distance $\cW_{2}(\bP_{x_0} \times \bP_{U_t}, \bP_{x_0^{\prime}} \times \bP_{U_t^{\prime}})$, which provides insight into designing output noise. The following theorem illustrates how to construct noise to protect the privacy of the probability distribution of the input.
\begin{secthm}\label{thm:dp}
Given $t \in \bZ_{+}$, $c > 0$, and $\delta \in (0, 1)$, the Gaussian mechanism~\eqref{eq:mech} induced by the white noise $V_t \sim \cN_{(t+1)q}(0, \Sigma_{V_t})$ is $(0,\delta)$ differentially private for $\operatorname{Adj}_2^c$ at $t$ if $\Sigma_{V_t} \succ 0$ is chosen as
\begin{align}\label{eq1:dp}
\min\left(\lambda_{\min}\left(\Sigma_{Y_{v,t}}\right), \lambda_{\min}\left(\Sigma_{Y_{v,t}'}\right)\right) \geq \frac{c^2\lambda_{\max}(\mathcal{O}_t)}{2 \delta^2},
\end{align}
where $\Sigma_{Y_{v,t}} = O_t \Sigma_{x_0} O_t^{\top} + N_t \Sigma_{U_t} N_t^{\top} + \Sigma_{V_t}$ , $\Sigma_{Y_{v,t}'} = O_t \Sigma_{x_0'} O_t^{\top} + N_t \Sigma_{U_t'} N_t^{\top} + \Sigma_{V_t}$ and
\begin{align}
\mathcal{O}_t:=\begin{bmatrix}
    O_t & N_t
\end{bmatrix}^{\top}\begin{bmatrix}
    O_t & N_t
\end{bmatrix}.
\end{align}
\end{secthm}
\begin{proof}
    The proof is reported in Appendix~\ref{app3}.
\end{proof}

We now give several remarks regarding Theorem~\ref{thm:dp}. First, only the matrix $\mathcal{O}_t$ depends on the system dynamics \eqref{eq:sys}. In fact, $\mathcal{O}_t$ measures the sensitivity of the system outputs to inputs. Second, condition \eqref{eq1:dp} involves the distribution $\bP_{Y_{v,t}}$ and $\bP_{Y_{v,t}'}$, meaning that the privacy performance is not uniform in $\mathrm{Adj}_2^c$; to see this, it is easy to check that the same $\Sigma_{V_t}$ leads to a smaller $\delta$ if the covariances in the pair $(\Sigma_{x_0} \times \Sigma_{U_t}, \Sigma_{x_0'} \times \Sigma_{U_t'})$ are larger. Nevertheless, the following corollary provides a way to design $(0,\delta)$ differential privacy for any pair in $\mathrm{Adj}_2^c$.
\begin{seccor}
\label{cor:dp}
    Given $t \in \bZ_{+}$, $c > 0$, and $\delta \in (0, 1)$, the Gaussian mechanism~\eqref{eq:mech} induced by the white noise $V_t \sim \cN_{(t+1)q}(0, \Sigma_{V_t})$ is $(0,\delta)$ differentially private for $\operatorname{Adj}_2^c$ at $t$ if $\Sigma_{V_t} \succ 0$ is chosen as
\begin{align}\label{eq3:dp}
\lambda_{\min}(\Sigma_{V_t}) \geq \frac{c^2\lambda_{\max}(\mathcal{O}_t)}{2 \delta^2}
\end{align}
\end{seccor}
\begin{proof}
    The proof follows directly from Theorem~\ref{thm:dp} by noticing  that $\min\left(\lambda_{\min}\left(\Sigma_{Y_{v,t}}\right), \lambda_{\min}\left(\Sigma_{Y_{v,t}'}\right)\right) \geq \lambda_{\min}(\Sigma_{V_t})$.
\end{proof}

It is worth noting that the covariance $\Sigma_{V_t}$ can become significantly large if $\delta$ is chosen to be very small. This occurs because no constraints are imposed on $\Sigma_{x_0}$ and $\Sigma_{U_t}$. In pathological cases, these two covariances can approach zero, necessitating a large $\Sigma_{V_t}$. If the distribution of $x_0$ is publicly known, we can establish the following theorem, which may lead to a smaller $\Sigma_{V_t}$.
\begin{secthm}
\label{thm:dp2}
    Consider that the distribution of $x_0$ is public information. Given $t \in \bZ_{+}$, $c > 0$, and $\delta \in (0, 1)$, the Gaussian mechanism~\eqref{eq:mech} induced by the white noise $V_t \sim \cN_{(t+1)q}(0, \sigma^2 I_{(t+1)q})$ is $(0,\delta)$ differentially private for $\operatorname{Adj}_2^c$ at $t$ if $\sigma > 0$ is chosen as
\begin{align}\label{eq1:dp}
\lambda_{\min}(O_t \Sigma_{x_0} O_t^{\top}) + \sigma^2 \geq \frac{c^2\lambda_{\max}(N_t^{\top}N_t)}{2 \delta^2}.
\end{align}
\end{secthm}
\begin{proof}
    By treating $O_t x_0$ the same as the noise $V_t$, the proof is similar to Theorem~\ref{thm:dp} and is thus omitted.
\end{proof}

Theorem~\ref{thm:dp2} introduces a specific condition that constrains $\mathrm{Adj}_2^c$ to reduce the conservatism of Corollary~\ref{cor:dp}.  Furthermore, it is possible to have more efficient Gaussian mechanisms by imposing additional constraints on $\mathrm{Adj}_2^c$.

\begin{secrem}
    In the classical Gaussian mechanism (see e.g., \cite{dwork2006calibrating}, \cite{le2013differentially}), $\epsilon > 0$ is a necessary condition. However, in \cite{balle2018improving}, the authors provided sufficient condition for $(0, \delta)$-differential privacy with the Gaussian mechanism. The results presented in this paper generalize Theorem 2 in \cite{balle2018improving} to cases where statistical inputs are employed.
\end{secrem}

\section{Numerical Experiments}
\label{sec:sim}
In this section, we revisit the motivation example in Section~\ref{sec:pre} by numerical experiments. The system model is assumed to evolve according to
\begin{align*}
    x(t+1) & =0.9 x(t)+ u(t), \\
y(t) & =x(t)+v(t)
\end{align*}

We assume the distribution of $x_0$ is public and $x_0 \sim \cN_1(90, 10)$. Moreover, we set $\Bar{m} = 20$, $m_u = 1$ and $\Sigma_u = 0.1$. 

To design the noise, we apply Theorem~\ref{thm:dp2} with $ t = 2 $. By setting $ c = 2.02 $, the pair of input probability distributions parameterized by one individual and by two individual is in $\mathrm{Adj}_2^c$. Next, the value of $ \sigma $ in Theorem~\ref{thm:dp2} is calculated as $ \sigma = 13.9193 $ with $ \delta = 0.1 $. To examine the impact of privacy on performance, we also set $ \delta = 0.2 $ and compute $ \sigma = 6.9596 $ according to Theorem ~\ref{thm:dp2}. The performance of the sensor outputs is illustrated in Figure~\ref{fig:1} and Figure \ref{fig:2}. It can be observed that the outputs for one and two individuals are statistically similar, making it hard for an eavesdropper to accurately distinguish the occupancy level from three observations. Additionally, the output performance improves as privacy protection weakens, i.e., when $\delta$ is larger.

\begin{figure}
    \centering
    \includegraphics[width=1\linewidth]{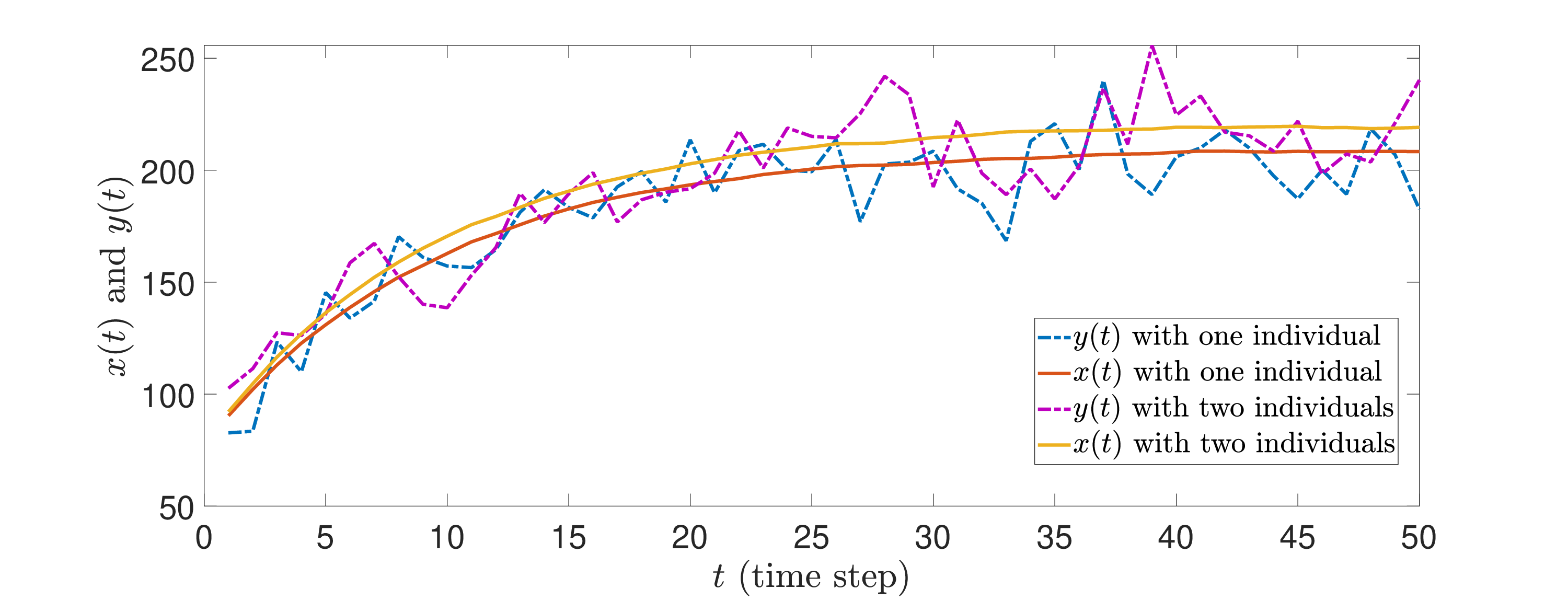}
    \caption{Output Performance with $\delta = 0.1$}
    \label{fig:1}
\end{figure}

\begin{figure}
    \centering
    \includegraphics[width=01\linewidth]{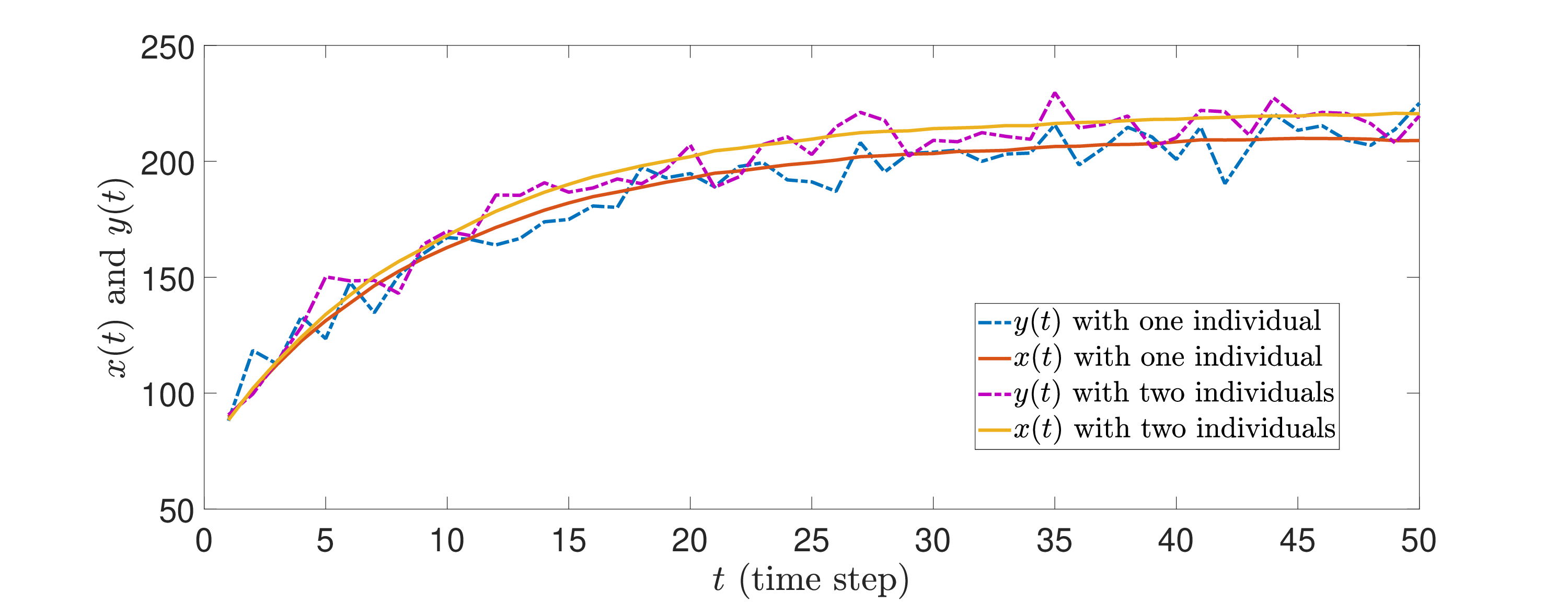}
    \caption{Output Performance with $\delta = 0.2$}
    \label{fig:2}
\end{figure}

\section{Conclusion}
\label{sec:con}
In this work, we formulated and analyzed the problem of privacy preservation for statistical inputs in linear time-invariant systems. By introducing an adjacency relationship and leveraging key relationships between statistical divergences, we established theoretical foundations for designing privacy-preserving Gaussian mechanisms. Future work may explore extensions to privacy protection of system parameters.
\appendices

\section{Proof of Lemma~\ref{lem:disinq}}
\label{app1}
To begin with, we first verify the following technical lemma, which is important in the proof of Lemma \ref{lem:disinq}.
\begin{seclem}
\label{lem:app1}
    For $ \lambda > 0$, it holds that 
    \begin{align*}
        \lambda^2 + \frac{1}{\lambda^2} - 2 \leq  2(\lambda - 1)^2 + 2(\frac{1}{\lambda} - 1)^2.
    \end{align*}
\end{seclem}
\begin{proof}
    Expanding the inequality yields
    \begin{align*}
         \lambda^2 + \frac{1}{\lambda^2} - 2 \leq  2(\lambda^2 + \frac{1}{\lambda^2} - 2\lambda - \frac{2}{\lambda}) + 4.
    \end{align*}
By defining $z:= \lambda+\frac{1}{\lambda}$, one can rewrite the given inequality as
\begin{align*}
    z^2 - 4 \leq 2(z^2 -2z),
\end{align*}
which is equivalent to
\begin{align*}
    -(z - 2)^2 \leq 0.
\end{align*}
Since the left-hand side is always non-positive, the inequality holds for all 
$z$. This concludes the proof.
\end{proof}

Now, we proceed to prove Lemma \ref{lem:disinq}. Recall the closed expression of the KL divergence for two Gaussian distributions as follows,
\begin{align*}
    \kl(\bP_1 \| \bP_2)  = &\frac{1}{2} [(m_2 - m_1)^{\top} \Sigma_2^{-1} (m_2 - m_1)+ \tr(\Sigma_2^{-1} \Sigma_1)   \\
    &- n + \log \frac{\det \Sigma_2}{\det \Sigma_1}].
\end{align*}
Then, one can compute the symmetrized KL divergence as follows,
\begin{align*}
    &\kl(\bP_1 \| \bP_2) + \kl(\bP_2 \| \bP_1)  \\
    = &\frac{1}{2} [(m_2 - m_1)^{\top} (\Sigma_1^{-1} + \Sigma^{-1}_2) (m_2 - m_1)+ \tr(\Sigma_2^{-1} \Sigma_1)   \\
    &+ \tr(\Sigma_1^{-1} \Sigma_2) - 2n ].
\end{align*}
Next, we construct upper bounds for each term step by step.
First, we show the upper bound of the term $(m_2 - m_1)^{\top} (\Sigma_1^{-1} + \Sigma_2^{-1}) (m_2 - m_1)$. By Rayleigh’s Entropy Theorem \cite{horn2012matrix}, it yields that
\begin{align}
\label{eq:mean}
    &(m_2 - m_1)^{\top} (\Sigma_1^{-1} + \Sigma_2^{-1}) (m_2 - m_1) \nonumber \\
    \leq & (\frac{1}{\lambda_{\min} (\Sigma_1)} + \frac{1}{\lambda_{\min} (\Sigma_2)}) |m_2 - m_1|_2^2 \nonumber \\
    \leq & \frac{2}{\min (\lambda_{\min} (\Sigma_1), \lambda_{\min} (\Sigma_2))} |m_2 - m_1|_2^2.
\end{align}
Define $\bS_{++}^n \ni X : =\Sigma_2^{-\frac{1}{2}} \Sigma_1 \Sigma_2^{-\frac{1}{2}}$. We have 
\begin{align*}
    \tr(\Sigma_2^{-1} \Sigma_1) +\tr(\Sigma_1^{-1} \Sigma_2) = \sum_{i = 0}^n (\lambda_{i}(X) + \frac{1}{\lambda_{i}(X)}),
\end{align*}
where $\lambda_{i}(X)$ is $i$-th smallest eigenvalue of $X$. Since $X$ is positive definite, we have $\lambda_{i}(X) > 0$. Using Lemma~\ref{lem:app1} gives 
\begin{align}
\label{eq:ineq_app1_1}
    &\tr(\Sigma_2^{-1} \Sigma_1) +\tr(\Sigma_1^{-1} \Sigma_2) - 2n \nonumber\\
    \leq & \sum_{i = 1}^n 2(\sqrt{\lambda_i(X)} -1)^2 + 2(\frac{1}{\sqrt{\lambda_i(X)}} - 1)^2. 
\end{align}
The covariance term in 2-Wasserstein distance can be rewritten as
\begin{align*}
    &\tr ( \Sigma_1 + \Sigma_2 - 2 (\Sigma_1^{1/2} \Sigma_2 \Sigma_1^{1/2})^{1/2} ) \\
    =& \frac{1}{2} \tr (\Sigma_2 (X^{\frac{1}{2}} - I_n)^2) + \frac{1}{2} \tr (\Sigma_1 (X^{-\frac{1}{2}} - I_n)^2) \\
    \geq & \frac{1}{2}  \lambda_{\min}(\Sigma_2)\tr ( (X^{\frac{1}{2}} - I_n)^2) + \frac{1}{2} \lambda_{\min}(\Sigma_1) \tr ( (X^{-\frac{1}{2}} - I_n)^2) \\
    \geq & \frac{1}{2} \min (\lambda_{\min}(\Sigma_1), \lambda_{\min}(\Sigma_2)) \\
    & \times \left(\tr \left( (X^{\frac{1}{2}} - I_n)^2\right) + \tr \left( (X^{-\frac{1}{2}} - I_n)^2\right)\right) 
\end{align*}
Further incorporating~\eqref{eq:ineq_app1_1} yields,
\begin{align*}
    & \min (\lambda_{\min}(\Sigma_1), \lambda_{\min}(\Sigma_2)) \\
    & \times \left(\tr \left( (X^{\frac{1}{2}} - I_n)^2\right) + \tr \left( (X^{-\frac{1}{2}} - I_n)^2\right)\right) \\
=&   \min (\lambda_{\min}(\Sigma_1), \lambda_{\min}(\Sigma_2)) \\ 
& \times \left(\sum_{i = 1}^n (\sqrt{\lambda_i(X)} -1)^2 + (\frac{1}{\sqrt{\lambda_i(X)}} - 1)^2\right) \\
\geq & \frac{1}{2} \min (\lambda_{\min}(\Sigma_1), \lambda_{\min}(\Sigma_2)) \\ 
& \times \left(\tr(\Sigma_2^{-1} \Sigma_1) +\tr(\Sigma_1^{-1} \Sigma_2) - 2n\right).
\end{align*}
This further implies that
\begin{align}
\label{eq:ineq_app1_2}
    &\tr ( \Sigma_1 + \Sigma_2 - 2 (\Sigma_1^{1/2} \Sigma_2 \Sigma_1^{1/2})^{1/2} ) \nonumber\\
    \geq& \frac{1}{4} \min (\lambda_{\min}(\Sigma_1), \lambda_{\min}(\Sigma_2)) \nonumber\\ &\left(\tr(\Sigma_2^{-1} \Sigma_1) +\tr(\Sigma_1^{-1} \Sigma_2) - 2n\right).
\end{align}
Combining~\eqref{eq:mean} and~\eqref{eq:ineq_app1_2}, we have
\begin{align*}
    &\kl(\bP_1 \| \bP_2) +\kl(\bP_2 \| \bP_1) \\
    \leq & \frac{2}{\min\left(\lambda_{\min}\left(\Sigma_1\right), \lambda_{\min}\left(\Sigma_2\right)\right)} \cW_2^2(\bP_1, \bP_2),
\end{align*}
which completes the proof.\red

\section{Proof of Lemma~\ref{lem:aff}}
\label{app2}
 Let $\gamma^* \in \Gamma(\mathbb{P}_x, \mathbb{P}_{x'})$ be the optimal coupling for $\mathcal{W}_2(\mathbb{P}_x, \mathbb{P}_{x'})$, i.e.,
\begin{align*}
        \mathcal{W}_2^2\left(\mathbb{P}_x, \mathbb{P}_{x'}\right) = \mathbb{E}_{(x, x') \sim \gamma^*} \left[ |x - x'|_2^2 \right].
\end{align*}
Define the coupling $\gamma$ for $(y, y')$ as:
\begin{align*}
    y = Fx + v, \quad y' = Fx' + v,
\end{align*}
    where $(x, x') \sim \gamma^*$ and $v \sim \mathbb{P}_v$. Then, we can verify that
\begin{align*}
    &\int_{\bR^m \times \bR^m} \mathbf{1}_{S \times \bR^m }(y, y') d \gamma\left(y, y^{\prime}\right) \\
    =&\int_{\bR^m}\int_{\bR^n \times \bR^n} \mathbf{1}_{S}(F x+v) \mathbf{1}_{\bR^m}(Fx'+v)d \gamma^*\left(x, x^{\prime}\right) d \mathbb{P}_v(v)\\
    =&\int_{\bR^m}\int_{\bR^n \times \bR^n} \mathbf{1}_{S}(F x+v) d \gamma^*\left(x, x^{\prime}\right) d \mathbb{P}_v(v), \forall S \in \mathcal{B}(\mathbb{R}^m)
\end{align*}

Since $\gamma^{*}$ is a valid coupling and $v$ is independent of $x$, we have
\begin{align*}&\int_{\bR^m}\int_{\bR^n \times \bR^n} \mathbf{1}_{S}(F x+v) d \gamma^*\left(x, x^{\prime}\right) d \mathbb{P}_v(v) \\
    =&\int_{\mathbb{R}^m} \int_{\mathbb{R}^n} \mathbf{1}_{S}(F x+v) d \mathbb{P}_x(x) d \mathbb{P}_v(v) =\mathbb{P}_y\left(S\right).
\end{align*}

Similarly, it holds that

\begin{align*}
    \int_{\bR^m \times \bR^m} \mathbf{1}_{\bR^m \times S }(y, y') d \gamma\left(y, y^{\prime}\right) = \mathbb{P}_{y^{\prime}}\left(S\right),
\end{align*}
which shows that $\gamma$ is a valid coupling.

Then, it can be calculated that
\begin{align*}
    \mathbb{E}_{(y, y') \sim \gamma} \left[ |y - y'|_2^2 \right] = \mathbb{E}_{(x, x') \sim \gamma^*}  \left[ |F(x - x')|_2^2 \right].
\end{align*}
By the definition of $|F|_2$, we have
\begin{align*}
    \mathbb{E}_{(y, y') \sim \gamma} \left[ |y - y'|_2^2 \right] \leq |F|_2^2 \cdot \mathbb{E}_{(x, x') \sim \gamma^*} \left[ |x - x'|_2^2 \right].
\end{align*}
Recall that the squared 2-Wasserstein distance is defined by
\begin{align*}
    \mathcal{W}_2^2\left(\mathbb{P}_y, \mathbb{P}_{y'}\right) = \inf_{\gamma \in \Gamma(\mathbb{P}_y, \mathbb{P}_{y'})} \mathbb{E}_{(y, y') \sim \gamma} \left[ |y - y'|_2^2 \right],
\end{align*} 
 it yields that
\begin{align*}
    \mathcal{W}_2^2\left(\mathbb{P}_y, \mathbb{P}_{y'}\right) &\leq \mathbb{E}_{(y, y') \sim \gamma} \left[ |y - y'|_2^2 \right]\\
    &\leq |F|_2^2 \mathbb{E}_{(x, x') \sim \gamma^*} \left[ |x - x'|_2^2 \right] \\
    &= |F|_2^2 \mathcal{W}_2^2\left(\mathbb{P}_x, \mathbb{P}_{x'}\right).
\end{align*}
This completes the proof. \red

\section{Proof of Theorem~\ref{thm:dp}}
\label{app3}
Let $\bP_{Y_{v,t}}$ denote the probability distribution of $Y_{v,t}$, $\bP_{Y_{v,t}'}$ denote the probability distribution of $Y_{v,t}$. The covariances of two Gaussian distributions can be computed as $\Sigma_{Y_{v,t}} = O_t \Sigma_{x_0} O_t^{\top} + N_t \Sigma_{U_t} N_t^{\top} + \Sigma_{V_t}$ and $\Sigma_{Y_{v,t}'} = O_t \Sigma_{x_0'} O_t^{\top} + N_t \Sigma_{U_t'} N_t^{\top} + \Sigma_{V_t}$.
By Lemma~\ref{lem:aff} and Definition~\ref{def:adj}, we have
\begin{align*}
    \cW_2^2(\bP_{Y_{v,t}}, \bP_{Y_{v,t}'}) \leq \lambda_{\max}(\mathcal{O}_t)c^2.
\end{align*}
Since $\bP_{Y_{v,t}}$ and  $\bP_{Y_{v,t}'}$ follow Gaussian distributions, it yields from Lemma~\ref{lem:disinq} that 
\begin{align*}
    &\kl(\bP_{Y_{v,t}} \| \bP_{Y_{v,t}'}) + \kl(\bP_{Y_{v,t}'} \| \bP_{Y_{v,t}}) \\
    \leq& \frac{2}{\min\left(\lambda_{\min}\left(\Sigma_{Y_{v,t}}\right), \lambda_{\min}\left(\Sigma_{Y_{v,t}'}\right)\right)} \cW_2^2(\bP_{Y_{v,t}}, \bP_{Y_{v,t}'}) \\
    \leq& \frac{2\lambda_{\max}(\mathcal{O}_t)c^2}{\min\left(\lambda_{\min}\left(\Sigma_{Y_{v,t}}\right), \lambda_{\min}\left(\Sigma_{Y_{v,t}'}\right)\right)}.
\end{align*}
From Pinsker's inequality \cite{cover1999elements}, we have
\begin{align*}
    &\mathrm{TV}(\bP_{Y_{v,t}}, \bP_{Y_{v,t}'}) \\
    \leq& \frac{1}{2}\sqrt{\kl(\bP_{Y_{v,t}} \| \bP_{Y_{v,t}'}) + \kl(\bP_{Y_{v,t}'} \| \bP_{Y_{v,t}}) },
\end{align*}
where $ \mathrm{TV}(\bP_{Y_{v,t}}, \bP_{Y_{v,t}'})$ is the total variation between $\bP_{Y_{v,t}}$ and $\bP_{Y_{v,t}'}$.

From the definition of total variation, it follows immediately that the mechanism~\eqref{eq:mech} is $(0, \delta)$ differentially private if $\mathrm{TV}(\bP_{Y_{v,t}}, \bP_{Y_{v,t}'}) \leq \delta$, which is ensured by~\eqref{eq1:dp}. \red

\bibliographystyle{ieeetr}
\bibliography{ref}
\end{document}